\documentclass[letterpaper, 10 pt, conference]{ieeeconf}  

\IEEEoverridecommandlockouts                             

\usepackage{graphics}
\usepackage[english]{babel}
\usepackage{epsfig} 
\usepackage{mathptmx} 
\usepackage{times} 
\usepackage{amsmath,mathtools} \usepackage{mathtools}
\usepackage{xcolor}
\DeclarePairedDelimiter\ceil{\lceil}{\rceil}
\DeclarePairedDelimiter\floor{\lfloor}{\rfloor}
\usepackage{tikz}
\allowdisplaybreaks
\usepackage{array}
\usepackage{mdwmath}
\usepackage{amssymb} 
\newtheorem{theorem}{Theorem}
\newtheorem{assumption}{Assumption}

\newtheorem{lemma}{Lemma}
\newtheorem{definition}{Definition}
\usetikzlibrary{arrows}
\usetikzlibrary{calc}

\newcommand{\comment}[1]{{ }}

\title{\LARGE \bf
Finite-sample analysis of rotation operator under $l_2$ norm and $l_\infty$ norm
}

\author{Mi Zhou $^{1}$   
\thanks{$^{1}$  Mi Zhou is with the School of Electrical and Computer Engineering, Georgia Institute of Technology, Atlanta, GA 30332.
({\tt\small mzhou91@gatech.edu}
)
        }%
}

\begin{document}
\maketitle
\thispagestyle{empty}
\pagestyle{empty}
\begin{abstract}
In this article, we consider a special operator called the two-dimensional rotation operator and analyze its convergence and finite-sample bounds under the $l_2$ norm and $l_\infty$ norm with constant step size.
We then consider the same problem with stochastic noise with affine variance.
Furthermore, simulations are provided to illustrate our results.
Finally, we conclude this article by proposing some possible future extensions.
\end{abstract}

\section{INTRODUCTION}
Looking for the fixed points of non-expansive mapping (i.e., $T(x)=x$) is an important topic in nonlinear mapping theory and has applications in image recovery and signal processing.
A myriad of research has been done on the properties and theorems of non-expansive operators.
While Banach fixed point theorem stated the existence and uniqueness of fixed point under a contractive mapping, the fixed-point set of a nonexpansive operator can be empty or contains multiple points.
As a direct consequence of non-expansiveness, it is not enough to directly iterate the operator $T$ to find a fixed point. 
Instead, one may iterate using the averaged operator $T_\alpha = (1-\alpha ) I + \alpha T$. 
Such iteration is also known as the Krasnosel'skii-Mann (KM) iteration~\cite{krasnosel1955two} and the update rule is given as in the following
\begin{equation*}
    x_{k+1} = (1-\alpha_k) x_k + \alpha_k T(x_k),
\end{equation*}
where $\{\alpha_k\}$ is the step size sequence.
Convergence of $x_k$ to a fixed-point was proved in~\cite{krasnosel1955two, schaefer1957methode, ishikawa1976fixed} under the bounded orbit assumption.
Under the non-empty fixed-point set assumption, the convergence result is analyzed in~\cite{ryu2016primer, zaiwei}.
The optimal convergence rate $O(1 / {\sqrt{k}})$ is obtained in~\cite{baillonrate,split,Optimaltransp,Bernoullis} for arbitrary norm.

Despite above works in deterministic case, the work in KM iteration of non-expansive operators with stochastic noise is sparse.
In \cite{zaiwei}, the authors derived a relaxed finite-sample bounds for non-expansive operators under $l_2$ norm with bounded variance.
However, the rotation map with some specific rotation angles under $l_\infty$ is neither contractive nor non-expansive, which makes all the existing works inapplicable.
Furthermore, the work on finite-sample analysis of non-expansive operators with affine noise is lack.
In this work, we aim to analyze the properties and finite-sample bounds of two-dimensional rotation operators under both $l_2$ and $l_\infty$ norm with and without affine stochastic noise.
We expect this work can give some hindsight in future studies of finite-sample bound for general non-expansive operators with and without noise.

This paper is organized as follows: in Section \ref{sec:problem}, we first introduce some preliminaries in normed linear space and non-expansiveness. 
We then formulate our problem by constructing two KM iterations under $l_2$ norm and $l_\infty$ norm respectively. 
We analyze their finite-sample bound and provide rigorous theoretic proof for each case.
Then in Section \ref{sec:noise}, we consider the same problem with noise and analyze its finite-sample bound.
Section \ref{sec:simulation} is our simulation results to illustrate our theoretical results in Section \ref{sec:problem} and Section \ref{sec:noise}.
Finally, we conclude our article in Section \ref{sec:conclusion}.

\section{Problem formulated}\label{sec:problem}
In this section, we will first introduce some preliminaries and then formulate our problem under different norms.
\subsection{Preliminaries}
\begin{definition}[Normed space]
A norm on the vector space $V$ is a function $||\cdot||$ that assigns to each vector $v \in
 V$ a real number $||v||$ such that for $c$ a scalar and $u, v\in V$, the following hold:
 \begin{enumerate}
     \item $||u||\geq 0$ with equality hold if and only if $u=0$.
     \item $||cu||=|c| ||u||$.
     \item (Triangle Inequality) $||u+v|| \leq ||u||+||v||$.
 \end{enumerate}
\end{definition}
A vector space $V$, together with a norm $||\cdot||$ on the space $V$, is called normed space.
The distance between $u$ and $v$ is $d(u,v) = ||u-v||$.

\begin{definition}
Let $V$ be one of the standard spaces $\mathbb{R}^n$ and $p \geq 1$ is a real number. The $p$-norm of a vector in $V$ is defined by
\begin{equation*}
||z||_p = (\sum_{i=1}^n |z_i|^p)^{\frac{1}{p}}.    
\end{equation*}
Specifically, when $p=2$, we have the familiar $l_2$ norm.
The $l_\infty$ norm of a vector in $V$ is defined as
\begin{equation*}
||z||_{\infty} = \max \{|z_i|,\; i=1,\cdots, n\}.
\end{equation*}
\end{definition}
\begin{definition}[matrix norm]
For a matrix $A\in \mathbb{R}^{m\times n}$, the operator norm is defined as
\begin{equation*}
||A||_2^2 = \lambda_{\mathrm{max}}(A^\top A), \quad ||A||_\infty = \max_{1\leq i\leq m}\sum_{j=1}^n ||a_{ij}||,
\end{equation*}    
and the following inequality holds for the matrix norm:
\begin{equation}\label{eqn:matrixnorm}
||Ax|| \leq ||A||||x||, \;\forall x \in \mathbb{R}^n.
\end{equation}
\end{definition}
\begin{definition}[\cite{KHAN20081}]
Let $C$ be a nonempty subset of a real Banach space $X$ and $T$ a self-mapping of $C$. Denote $F(T)$ as the set of fixed points of $T$. 
The mapping $T$ is said to be
\begin{enumerate}
    \item non-expansive if $||T(x)-T(y)|| \leq ||x-y||$ for all $x,y \in C$.
    \item quasi-nonexpansive if $||Tx-p|| \leq ||x-p||$, for all $x\in C$ and $p \in F(T)$.
    \item asymptotically nonexpansive if $\exists \{ u_n\}\in [0, +\infty)$ with $\lim_{n\rightarrow \infty}u_n=0$ and $||T^n x-p|| \leq (1+u_n)||x-p||$, for all $x\in C$ and $p\in F(T)$ and $n=1,2,...$.
\end{enumerate}
\end{definition}

\begin{definition}[Krasnosel'skii-Mann (KM) iteration \cite{KM}]
Let $T: C\rightarrow C$ be a nonexpansive map defined on a closed convex domain $C$ in a Banach space $(X, || \cdot ||)$. 
The Krasnosel'skii-Mann iteration approximates a fixed point of $T$ by the sequential averaging process:
\begin{align}
x_{n+1} = (1-\alpha_{n})x_n+\alpha_{n} Tx_n.
\end{align}
where $x_0 \in C$ is an initial guess and $\alpha_n\in(0,1)$ is a given sequence of scalar step sizes.
\end{definition}
\begin{lemma}[\cite{split}]
Let $T: \mathcal{H}\rightarrow \mathcal{H}$ be nonexpansive and $\alpha>0$. 
Then, $T_\alpha:=(1-\alpha)I+\alpha T$ ($I$ is identity mapping) and $T$ have the same set of fixed points.
\end{lemma}
\begin{theorem}
The fixed point for rotation matrix ($\theta \in(0, 2\pi)$) is $[0,0]^\top$.
\end{theorem}
\begin{proof}
 The determinant of $R$ is $\mathrm{det}(R)=1$, which means $R$ is full-rank.
 To make $Rx=x$, the only solution is $[0,0]^\top$.
\end{proof}
\subsection{Problem Description: $l_2$ norm}
We consider the following KM iteration
\begin{align}\label{eqn:l2rotateKM}
x_{k+1} = (1-\alpha_k)x_k+\alpha_k R x_k \frac{||x_k||_2}{||Rx_k||_2}
\end{align}
where $R$ is a two-dimensional rotation matrix defined as
\begin{align*}
R(\theta) = \begin{bmatrix}
\cos \theta & -\sin \theta\\ 
\sin \theta & \cos \theta
\end{bmatrix},
\end{align*}
where $\theta \in (0,2\pi)$ is the counter-clockwise rotation angle.
Since the rotation operator $R$ will only rotate the vector $x_k$ but won't change the size of the vector, we know $\frac{||x_k||_2}{||Rx_k||_2}=1$.
Thus Eqn. \eqref{eqn:l2rotateKM} can also be written as
\begin{align}
x_{k+1} = (1-\alpha_k)x_k+\alpha_k R x_k.
\end{align}
\begin{theorem}
The rotational operator is a non-expansive operator under the $l_2$ norm.
\end{theorem}
\begin{proof}
The norm of $R$ is $||R||_2=\sqrt{\lambda_{\max} (R^\top R)}=\sqrt{\lambda_{\max} (I)}=1$.
Thus the rotation operator is non-expansive.
\end{proof}
\begin{theorem}
The finite sample bound for constant step size is
\begin{align}
||x_k-x^*|| \leq [(\alpha-1)^2+2\alpha(1-\alpha)\cos\theta+\alpha^2]^\frac{k-1}{2} D,
\end{align}
where $D = \mathrm{dist}(x_1,x^*)$.
Specifically, when $\theta=\pi$, $||x_k-x^*||\leq (1-2\alpha)^{k-1}D$.
\end{theorem}
\begin{proof}
Denote $x^*$ as the fixed point.
Using the KM iteration, we have
\begin{align}
x_k = \left(\prod_{i=k-1}^{1}P_i \right) x_1
\end{align}
where $P_i = (1-\alpha_i)I+\alpha_i R$, $D = \mathrm{dist}(x_1,x^*)$.
Thus, we have
\begin{align}\label{eqn:l2main}
||x_k-x^*||_2 = ||x_k||_2 = \left\|\left(\prod_{i=k-1}^{1}P_i \right) x_1\right\|_2 \leq \left\|\left(\prod_{i=k-1}^{1}P_i \right)\right\|_2 || x_1||_2
\end{align}
where $P_i = \begin{bmatrix}
1-\alpha_i+\alpha_i \cos \theta &  -\alpha_i \sin \theta\\ 
\alpha_i \sin \theta & 1-\alpha_i+\alpha_i \cos \theta
\end{bmatrix}$.
We further found that
\begin{align*}
P_k^\top P_k &= [(1-\alpha_k)I+\alpha_k R^\top][(1-\alpha_k)I+\alpha_k R] \\
& = (1-\alpha_k)^2 I +\alpha_k(1-\alpha_k)(R+R^\top) + \alpha_k^2 R^\top R\\
& =(1-2\alpha_k+2\alpha_k^2)I+2\alpha_k(1-\alpha_k)\cos\theta I \\
& = (1-2\alpha_k+2\alpha_k^2+2\alpha_k(1-\alpha_k)\cos\theta)I.
\end{align*}
In the last two steps, we used the properties of the rotation matrix, i.e., $R^\top R = I$ and $R+R^\top = 2\cos\theta I$.
If $\alpha_i =\alpha \in (0,1)$, Eqn. \eqref{eqn:l2main} satisfies
\begin{align*}
||x_k-x^*|| \leq \left\|\left(\prod_{i=k-1}^{1}P_i \right)\right\| D \leq \left[\sqrt{\lambda_{\max}(P_i^\top P_i)}\right]^{k-1} D= \\
[(\alpha-1)^2+2\alpha(1-\alpha)\cos\theta+\alpha^2]^\frac{k-1}{2} D
\end{align*}
i.e.,
\begin{align*}
||x_k-x^*|| \leq [(\alpha-1)^2+2\alpha(1-\alpha)\cos\theta+\alpha^2]^\frac{k-1}{2} D.
\end{align*}
Define $g(\alpha) = (\alpha-1)^2+2\alpha(1-\alpha)\cos\theta+\alpha^2$, 
which can be rewritten as $g(\alpha) = (2-2\cos\theta) \alpha^2+(2\cos\theta-2)\alpha+1$. 
Taking derivative $g'(\alpha) = 2(2-2\cos\theta)\alpha+(2\cos\theta-2)=0$, we thus have $\alpha = 0.5$, which means when $\alpha=0.5$, $g(\alpha)$ is minimal with value $\frac{1+\cos\theta}{2}$.
This also means when the step size is 0.5, we have the fastest convergence speed.

Actually, we can find the exact solution for this iteration based on the fact that
\begin{align*}
x_{k+1}^\top x_{k+1}=||x_{k+1}||^2 = x_k^\top P_k^\top P_k x_k \\
=(1-2\alpha_k+2\alpha_k^2+2\alpha_k(1-\alpha_k)\cos\theta)x_{k}^\top x_k.
\end{align*}
This is exactly the bound derived before.
The coefficient $ (1-2\alpha_k+2\alpha_k^2+2\alpha_k(1-\alpha_k)\cos\theta) < 1$ because $\theta \in (0,2\pi)$ and thus $\cos\theta < 1$.
\end{proof}
\subsection{Problem Description: $l_\infty$ norm}
In this section, we consider the following KM iteration:
\begin{align}\label{eqn:inf_KM}
x_{k+1} = (1-\alpha)x_k +\alpha R x_k \frac{||x_k||_{\infty}}{||Rx_k||_{\infty}}.
\end{align}
We aim to prove the convergence of this iteration and how fast it will converge.
Without specific clarification, the norm notation in this subsection all means $\infty$-norm.
The results of $\theta \in(\pi,2\pi)$ are the same as that of $(0,\pi)$, which can be regarded as rotating counter-clockwise with $\theta$, so we only prove the result for $\theta \in(0,\pi)$. 
\begin{lemma}
The rotation operator under $l_\infty$ norm, i.e., $R_{\infty, \theta} = R\frac{||x||}{||Rx||}$ is non-expansive only when $\theta\in \Theta $ where $\Theta = \{\frac{\pi}{2},\pi\}$.
\end{lemma}
\begin{proof}
When $\theta\in \Theta$, $\frac{||x_k||}{||Rx_k||}=1$ and $R_{\infty,\Theta}=R$.
Using the inequality in \eqref{eqn:matrixnorm}, we have $||R_{\infty,\Theta}x-R_{\infty,\Theta}y||=||Rx-Ry||\leq ||R||||x-y||$.
Since when $\theta\in \Theta$, $||R||= 1$. Thus, $R_{\infty,\Theta}$ is non-expansive.
A counter-example to show that $R_{\infty,\theta}$ for $\theta\notin \Theta$ is not non-expansive is easy to be found.
See the following figure with $\theta= \frac{\pi}{4}$.
It is obvious that $||R_{\infty,\theta}x-R_{\infty,\theta}y|| = ||[1,1]^\top-[-1,1]^\top|| = 2$ while $||x-y||=||[1,0]^\top-[0,1]^\top|| = 1$.
$||R_{\infty,\theta}x-R_{\infty,\theta}y||>||x-y||$.
So it is not non-expansive when $\theta=\frac{\pi}{4}$.
\begin{center}
\begin{tikzpicture}
\draw[thick,->] (0,0) -- (0,2)
node[anchor=east]{$y$};
\draw[thick,->] (-2,0) -- (2,0)
node[anchor=north]{$x$};
\draw[blue, thick,->] (0,0) -- (1,0)
node[anchor=north]{$1$}
node[anchor=south]{$x$}
node at (-1,-0.2) {$-1$};
\draw[red, thick,->] (0,0) -- (1,1)
node[anchor=west]{$R_{\infty,\theta}x$};
\draw[blue, thick,->] (0,0) -- (0,1)
node[anchor=west]{$1$}
node[anchor=east]{$y$};
\draw[red, thick,->] (0,0) -- (-1,1)
node[anchor=east]{$R_{\infty,\theta}y$};
\draw[gray, dashed] (1,0) -- (1,1);
\draw[gray, dashed] (0,1) -- (1,1);
\draw[gray, dashed] (-1,0) -- (-1,1);
\draw[gray, dashed] (0,1) -- (-1,1);
\end{tikzpicture}
\end{center}
\end{proof}
\begin{theorem}\label{thm:nonexpansive}
Denote $P_k = (1-\alpha)I+\alpha \frac{||x_k||_{\infty}}{||Rx_k||_{\infty}}R$.
The average operator $P_k$ is quasi-nonexpansive for all the $\theta \in (0,2\pi)$.
\end{theorem}
\begin{proof}
Using the triangle inequality of norm, we obtain
\begin{align*}
||x_{k+1}||=||P_k x_k|| &= ||\left((1-\alpha)I+\alpha \frac{||x_k||_{\infty}}{||Rx_k||_{\infty}}R \right) x_k|| \\
&\leq (1-\alpha)||x_k||+\alpha ||\frac{||x_k||_{\infty}}{||Rx_k||_{\infty}}Rx_k||\\
&=||x_k||.
\end{align*}
Using the definition in \cite{KHAN2016231}, we know that $P_k$ is a quasi-nonexpansive operator.
\end{proof}

\begin{theorem}
 The KM iteration \eqref{eqn:inf_KM} converges to the fixed point for all $\theta \in (0, \pi)$.  
\end{theorem}
\begin{proof}
Denote 
\begin{align}
\gamma(x_k) = \frac{||x_k||_\infty}{||Rx||_\infty}.
\end{align}
It is obvious that $ \gamma\in \left[\frac{\sqrt{2}}{2}, \sqrt{2}\right]$ by using some geometric intuition.
We can write \eqref{eqn:inf_KM} as 
\begin{align}
x_{k+1} = P_k x_k
\end{align}
where $P_k$ can be written as
\begin{align}
P_k(x_k) = \begin{bmatrix}
1-\alpha+\alpha \gamma\cos\theta & -\alpha \gamma \sin\theta\\ 
 \alpha \gamma \sin\theta& 1-\alpha+\alpha\gamma\cos\theta 
\end{bmatrix}.
\end{align}
which means that $P_k$ is a variable depending on $x_k$.
If $\theta =\pi$, we got $||P_k||_\infty = 1-2\alpha<1$ which means $P_k$ is contractive in this case.
If $\theta \notin \Theta$, using Theorem \ref{thm:nonexpansive}, we know it will converge to the fixed point.
\end{proof}

\begin{theorem}
The finite sample bound for $\theta \in \Theta$ is
\begin{align} \label{eqn:infpiover2}
||x_k-x^*||_\infty = ||x_k||_\infty \leq \left\{\begin{matrix}
(1-2\alpha)^{k-1} D, \; \theta = \pi\\ 
(0.5)^{\floor*{\frac{k-1}{2}}}D, \; \theta = \frac{\pi}{2}
\end{matrix}\right. 
\end{align}
\end{theorem}
\begin{proof}
We prove this theorem separately.
If $\theta \in \Theta$, for example, suppose $\theta = \pi$.
We have a rotation matrix
\begin{align}
R = \begin{bmatrix}
-1 & 0\\ 
0 & -1
\end{bmatrix}
\end{align}
and the property $\frac{||x_k||}{||Rx_k||}=1$, which makes the KM iteration the following
\begin{align}
x_{k+1} = ((1-\alpha)I+\alpha R)x_k\\
i.e., \; x_{k+1} = \begin{bmatrix}
1-2\alpha & 0\\ 
0 & 1-2\alpha
\end{bmatrix} x_k.
\end{align}
This will lead to
\begin{align}
x_k = (\underbrace{\begin{bmatrix}
1-2\alpha & 0\\ 
0 & 1-2\alpha
\end{bmatrix}}_P)^{k-1} x_1
\end{align}
thus,
\begin{align}
||x_k|| \leq ||P^{k-1}||||x_1||= (1-2\alpha)^{k-1} D
\end{align}
which means geometric convergence speed.

If $\theta = \frac{\pi}{2}$, we have
\begin{align}
x_k = (\underbrace{\begin{bmatrix}
1-\alpha & -\alpha\\ 
\alpha & 1-\alpha
\end{bmatrix}}_P)^k x_1.
\end{align}
We find that the optimal infinity norm of $||P^2||_{\infty} = \frac{1}{2}$ happens when $\alpha=\frac{1}{2}$.
Using the same logic, 
\begin{align}
||x_k|| \leq ||P^{k-1}|| ||x_1|| \leq ||P^2||^{\floor*{\frac{k-1}{2}}} D = (0.5)^{\floor*{\frac{k-1}{2}}}D.
\end{align}
The result for $\theta=\frac{3\pi}{2}$ is the same as that of $\theta = \frac{\pi}{2}$.
\end{proof}
Now we consider the same problem when $\theta \notin \Theta$.
The difficulty here lies in $\gamma$ is not a constant but a variable.
So we will write $P_k$ as $P_k(x_k)$.
First, let's see $\theta = \frac{3\pi}{4}$, we found that $||P_k(x_k)||_\infty = 0.5$ when $\alpha=0.5$.
This is obvious since $ ||P_k|| = |1-\alpha+\alpha \gamma(x_k)\cos(3\pi/4)| + |\alpha \gamma(x_k)\sin(3\pi/4)|$.
If $\alpha = 0.5$, this value is $1-\alpha-\alpha\gamma(x_k) \sqrt{2}/2+\alpha \gamma(x_k)\sqrt{2}/2=1-\alpha = 0.5$.
Thus,
\begin{align*}
||x_k|| \leq ||P^{k-1}||||x_1|| \leq ||P||^{k-1} D=(0.5)^{k-1} D.
\end{align*}

\begin{theorem} \label{thm:smallthetabd}
Fix the step size $\alpha=0.5$. 
Denote $\theta = \frac{p}{q}\pi$ and $p<q$ are integers.
The finite sample bound for $\theta \notin \Theta$ is
\begin{align}\label{eqn:infpiover4}
||x_k-x^*||_\infty \leq 
\left\{\begin{matrix}
\beta_u^{\floor*{(k-1)/T}} D,\; \theta\in (0,\pi/2)\\ 
\left(\frac{1+\tan(3\pi/4-\theta)}{2}\right)^{(k-1)} D,\; \theta\in(\pi/2,\pi)
\end{matrix}\right.
\end{align}
where $T = \ceil*{1/\frac{p}{q}}$ is called the pseudo-period for $\theta = \frac{p}{q}\pi$. 
$\beta_u < 1$ is a bound value dependent on the value of $\theta$.
\end{theorem}
\begin{lemma}
For any $i>0$ and $\theta \in(0,\pi/2)$ with corresponding period $T$, the following condition holds:
\begin{align}
 \frac{1}{2}\left(1+\tan\left(\frac{\pi}{4}-\frac{\theta}{2}\right)\right)||x_i|| \leq ||x_{i+T}|| =|| P_{i+T} \cdots P_{i} x_i||_\infty, \\
 ||x_{i+T}|| =|| P_{i+T} \cdots P_{i} x_i||_\infty \leq \beta_u||x_i||.
\end{align}
\end{lemma}
\begin{proof}
 We will prove it from the geometric perspective.
Firstly, the lower bound can be derived from the following figure.
\begin{center}
\begin{tikzpicture}
\draw[thick,->] (0,-4) -- (0,4)
node[anchor=east]{$y$};
\draw[thick,->] (-4,0) -- (4,0)
node[anchor=north]{$x$};
\draw[black, thick] (-3,-3) rectangle (3,3);
\draw[blue, thick,->] (0,0) -- (3,0)
node at (2.5,-0.2) {$x_i$};
\draw[blue, thick,->] (0,0) -- (3,1.5)
node at (2.5,0.6) {$x_{i+T-1}$};
\filldraw[black] (3,0) circle (2pt) node[anchor=west]{A};
\filldraw[black] (3,1.5) circle (2pt) node[anchor=west]{B};
\filldraw[black] (3,2.25) circle (2pt) node[anchor=west]{C};
\filldraw[black] (3,3) circle (2pt) node[anchor=west]{D};
\draw[blue, thick,->] (0,0) -- (1.5,3)
node[anchor=south]{$Rx_{i+T-1}\gamma(x_{i+T-1})$};
\draw[gray, thick, dotted] (0,0) -- (3,3);
\draw[gray, thick] (3,1.5) -- (1.5,3);
\draw[blue, thick,->] (0,0) -- (2.25,2.25)
node at (2,1.5) {$x_{i+T}$};
\filldraw[black] (2.25,2.25) circle (2pt) node at (1.8, 2.25) {E};
\draw[->,>=stealth',semithick] (25.6:1.2cm) arc[radius=1.2, start angle=26, end angle=65];
\node at (1,1) {$\theta$};
\draw[red, thick,dotted] (2.25,2.25) -- (3,2.25);
\end{tikzpicture}
\end{center}
Denote $||x_i|| = r$.
As we can see, the \textbf{best} progress this iteration can make happens when the line connecting $x_{i+T-1}$ and $Rx_{t+T-1}\gamma(x_{i+T-1})$ is perpendicular to the gray line, which gives us $|CD| = 1/2|BD| = 1/2(r-|AB|)$,
and thus,
\begin{align*}
||x_i||-||x_{i+T}\| = |CD| \geq \frac{1}{2}\left[1-\tan\left(\frac{\pi}{4}-\frac{\theta}{2}\right)\right]||x_i||\\
\Rightarrow ||x_{i+T}|| \geq  \frac{1}{2}\left[1+\tan\left(\frac{\pi}{4}-\frac{\theta}{2}\right)\right]||x_i||.
\end{align*}
To make the notation simple, we denote $\beta_l = \frac{1}{2}\left[1+\tan\left(\frac{\pi}{4}-\frac{\theta}{2}\right)\right]$.

Secondly, the upper bound $\beta_u$ depends on the rotation angle $\theta$.
Regarding the KM iteration in \eqref{eqn:inf_KM} as a geometric process in the above square, we can imagine a vector starts from an initial position $x_k$, then rotates with angle $\theta$ to obtain $Rx_k$, next scales to map to the square $Rx_k\frac{||x_k||}{||Rx_k||}$, finally is taken average with the original vector $x_k$ to obtain the new vector $x_{k+1}$.
The vector will shrink into a new smaller square if the initial vector meets with the corner of the original square, which we call ``making progress''.
Imagining the extreme case, the initial vector starts from one corner of the above, this iteration will definitely make progress after $\frac{\pi/2}{1/2\theta}=\frac{\pi}{\theta}$ steps.
That's how we define the period $T$.
\begin{definition}[self-similarity]
The bound $\beta_u= \max\left\{\frac{||x_{i+T||}}{||x_i||}\right\}$ for KM iteration \eqref{eqn:inf_KM} with a fixed initial point $x_1$, is less than or equal to the $\beta_u=\max\left\{\frac{||x_{1+T||}}{||x_1||}\right\}$ where $x_1$ is taken with all the value in a lateral of the above square.
\end{definition}
As we know, for $\theta=\pi/2$, the $\beta_u$ is derived as $1/2$.
However, for the other angles, the $\beta_u$ is obtained by the numerical search.
Table \ref{tab:betau} shows their relationship.
Fixed $\theta$, we found that the value of $\frac{||x_{i+T||}}{||x_i||}$ depends on $||x_i||$.
That being said, if $||x_i||$ is confirmed, then $||x_{i+T}||$ is confirmed.
Due to the self-similarity of this iteration process, we thus only need to consider the bound of $\frac{||x_{1+T}||}{||x_1||}$ with $x_1$ starting from different positions in a lateral of the above square.
By doing a brute force search for $x_1$ in a lateral of the square and then due to the self-similarity, we find the bound for each angle with the precision of 0.0001.
\begin{table}[!htp]
    \centering
    \begin{tabular}{|c|c|} \hline
     $\theta$ & $\beta_u$ \\ \hline \hline 
     $\pi/12$ & 0.8974\\  \hline
      $\pi/6$ & 0.8211\\  \hline
      $\pi/4$ &0.7504  \\  \hline
      $\pi/3$ & 0.6830 \\  \hline
        $\pi/2$& 0.5 \\  \hline
    \end{tabular}
    \caption{Rotation angle and their corresponding bound.}
    \label{tab:betau}
\end{table}

Thus, we finally proved the conclusion:
\begin{align}
\beta_l ||x_i|| \leq ||x_{i+T}||\leq \beta_u ||x_i||.
\end{align}
\end{proof}

After lemma 2 is proved, we can then prove Theorem \ref{thm:smallthetabd} directly.
For $x_k =P_{k-1}  ...P_2 P_1  x_1$, we group the multiplication of $P_i$ with the period $T$ and using the non-expansiveness of $P_i$, then we can obtain the bound directly
\begin{align*}
||x_k|| \leq \beta_u^{\floor*{(k-1)/T}}D .
\end{align*}
For example, when $\theta = \pi/3$, the period is $T=3$.
For $k=8$,
\begin{align*}
||x_8|| \leq \beta_u ||x_5|| \leq \beta_u \beta_u ||x_2|| \leq \beta_u^2 ||x_1||
\end{align*}
where the last inequality used the quasi-nonexpansive property of $P_1$ that says $||x_2||=||P_1 x_1|| \leq ||x_1||$.

For $k=9$, we have
\begin{align*}
||x_9|| = ||P_8...P_1 x_1|| \leq \beta_u ||x_6|| \leq \beta_u\beta_u ||x_3|| \leq \beta_u^2 ||x_1||
\end{align*}
where the last step used the quasi-non-expansive property $||x_3|| = ||P_2x_2||\leq ||x_2||=||P_1x_1||\leq ||x_1||$.

Now, consider $\theta \in (\pi/2, \pi)$.
As we know, when $\theta >\pi/2$, every iteration will make progress.
\begin{theorem}
\begin{align}
||x_{i+1}|| \leq  \frac{1+\tan(3\pi/4-\theta)}{2}  ||x_i||.
\end{align}
\end{theorem}
\begin{proof}
This can be proved using geometry and is verified with simulation. 
In the following figure, $\psi=\pi-\theta-\pi/4 = 3\pi/4-\theta$. Denote $r=||x_i||$.
$|AC| = r \tan\psi$. $|BC|=(1-\tan\psi)r$.
The minimum jump is $|BE| = \frac{|BC|}{2}=(1-\tan\psi)r/2$. $E$ is the middle point of $BC$ because $FE \parallel DB$ and $F$ is the middle point of $CD$.
Thus,
\begin{align*}
||x_i|| - ||x_{i+1}|| &\geq (1-\tan\psi)r/2 \\
& \Rightarrow ||x_{i+1}|| \leq (1+\frac{1-\tan\psi}{2})r = \frac{1+\tan\psi}{2}||x_i||.
\end{align*}
\begin{center}
\begin{tikzpicture}
\draw[thick,->] (0,-4) -- (0,4)
node[anchor=east]{$x_2$};
\draw[thick,->] (-4,0) -- (4,0)
node[anchor=north]{$x_1$};
\draw[black, thick] (-3,-3) rectangle (3,3);
\draw[blue, thick,->] (0,0) -- (3,0.8038)
node  at (2.6,0.3) {$x_i$};
\draw[blue, thick,->] (0,0) -- (-3,3)
node[anchor=south]{$Rx_{i}\gamma(x_{i})$};
\draw[gray, thick,->] (3,0.8038) -- (-3,3);
\draw[blue, thick,->] (0,0) -- (0,1.9019)
node at (0.3, 1.5) {$x_{i+1}$};
\filldraw[black] (0,1.9019) circle (2pt) node[anchor=east]{F};
\draw[gray, dashed] (0,1.9019) -- (3, 1.9019)
node[anchor=west]{E};
\draw[->,>=stealth',semithick] (15:1.2cm) arc[radius=1.2, start angle=15, end angle=135];
\node at (1,1) {$\theta$};
\draw[->,>=stealth',semithick] (0:1.2cm) arc[radius=1.2, start angle=0, end angle=15];
\node at (1.5,0.2) {$\psi$};
\filldraw[black] (3,0) circle (2pt) node[anchor=west]{A};
\filldraw[black] (3,3) circle (2pt) node[anchor=west]{B};
\filldraw[black] (3,0.8038) circle (2pt) node[anchor=west]{C};
\filldraw[black] (-3,3) circle (2pt) node[anchor=east]{D};
\end{tikzpicture}
\end{center}
\end{proof}

\section{With noise} \label{sec:noise}
In this scenario, we follow the following assumption for noise:
\begin{assumption}\label{ass:noise}
Let $\mathcal{F}_k$ denote the $\sigma$-algebra generated by sequence $\{ x_0, \omega_0, x_1, \omega_1, \cdots,x_{k-1}, \omega_{k-1}, x_k\}$.
The noise sequence $\{ \omega_k\}$ satisfies for all $k\geq 0$,
\begin{enumerate}
    \item $\mathbb{E}[\omega_k|\mathcal{F}_k]=0$.
    \item $\mathbb{E}[\left \| \omega_k \right \|_c^2|\mathcal{F}_k] \leq A + B \left \| x_k \right \|_c^2$ for some constant $A>0$ and $B\geq 0$.
\end{enumerate}
where the subscript $c$ denotes $c$-norm.
\end{assumption}

\subsection{problem description: $l_2$ norm with noise}
The KM iteration for the rotation operator with noise is 
\begin{align}
x_{k+1} = x_k +\alpha_k(Rx_k-x_k +\omega_k) \nonumber\\
\mathrm{i.e.,} \; x_{k+1} = (1-\alpha_k)x_k +\alpha_k(Rx_k +\omega_k)
\end{align}
where $\omega_k$ is the noise sequence that satisfies Assumption \ref{ass:noise}.
We can write this system as
\begin{align*}
x_{k+1} = \begin{bmatrix}
1-\alpha_k+\alpha_k\cos\theta &-\alpha_k\sin\theta \\ 
\alpha_k\sin\theta & 1-\alpha_k+\alpha_k\cos\theta
\end{bmatrix}x_k+\alpha_k \omega_k
\end{align*}
\begin{theorem}
The finite sample bound for constant step size is
\begin{align}\label{eqn:2normnoise}
E[||x_k-x^*||_2^2] \leq (\mu+\alpha^2B)^{k-1} D^2 +A\alpha^2 \frac{1-(\mu+\alpha^2B)^{k-1}}{1-(\mu+\alpha^2B)} 
\end{align}
where $0<B < (1-\mu)/\alpha^2$.
\end{theorem}
\begin{proof}
As we already know, $P_k$ is non-expansive for all $\theta$.
\begin{align}
||x_{k+1}||^2 = x_k^\top P_k^\top P_k x_k+2\alpha_k\omega_k^\top P_k x_k +\alpha_k^2\omega_k^\top\omega_k
\end{align}
where
\begin{align*}
P_k^\top P_k = (1-2\alpha_k+2\alpha_k^2+2\alpha_k(1-\alpha_k)\cos\theta)I.
\end{align*}
Denote $\mu_k = (1-2\alpha_k+2\alpha_k^2+2\alpha_k(1-\alpha_k)\cos\theta)$ and we know $\mu_k<1$ holds for all $\theta\in(0,2\pi)$.
Thus,
\begin{align*}
&E[||x_{k+1} ||^2\vert \mathcal{F}_k]\\ 
&= x_k^\top P_k^\top P_k x_k + 2 \alpha_k E[\omega_k^\top \vert \mathcal{F}_k]P_k x_k + \alpha_k^2 E[\omega_k^\top \omega_k\vert \mathcal{F}_k]\\
&\leq \mu_k ||x_k||^2  +\alpha_k^2 E[\omega_k^\top \omega_k\vert \mathcal{F}_k]\\
&\leq (\mu_k+\alpha_k^2 B)||x_k||^2+A\alpha_k^2.
\end{align*}
Taking the total expectation and telescoping, we have
\begin{align*}
E[||x_{k+1}||^2] 
& \leq (\mu+\alpha^2B)^k||x_1||^2 +A\alpha^2 \frac{1-(\mu+\alpha^2B)^k}{1-(\mu+\alpha^2B)}. 
\end{align*}
supposing $\alpha_k=\alpha$.
If $B\neq 0$, to make the system converge, $\mu+B\alpha_k^2 < 1 \Rightarrow B<(1-\mu)/\alpha_k^2$.
\end{proof}

 
\section{Simulation}\label{sec:simulation}
In this section, we will provide some simulations to verify our above theoretical work one by one.
Our experiments are executed from the following aspects: (1) $l_2$ norm; (2) $l_\infty$ norm; (3) $l_2$ norm with noise; (4) $l_\infty$ with noise.
In all of the following settings, without specific claiming, the initial position is $x_0=[10, 30]^\top$.
\subsection{$l_2$ norm}
Different step sizes are used to test the convergence.
When $\alpha_k = \alpha$, that is constant step size, it converges to the fixed point with geometrical speed.
When $\alpha_k = \frac{1}{\log k}$ and $\alpha_k = \frac{1}{\sqrt{k}}$, it also converges to fixed point.
However, when $\alpha_k=\frac{1}{k}$ (i.e., diminishing step size), there is a constant error for $\theta \notin \Theta$.
Another phenomenon is the convergence speed of $\log$ step size is faster than that with the square root step size.
Fig.\ref{fig:l2} shows the trajectory and the value of $||x_k||_2$ using constant step size.
We can see that the system converges to its fixed point $[0,0]^\top$ geometrically fast as proved in Section \ref{sec:problem}.
\begin{figure}[!htp]
    \centering
    \includegraphics[width=\linewidth]{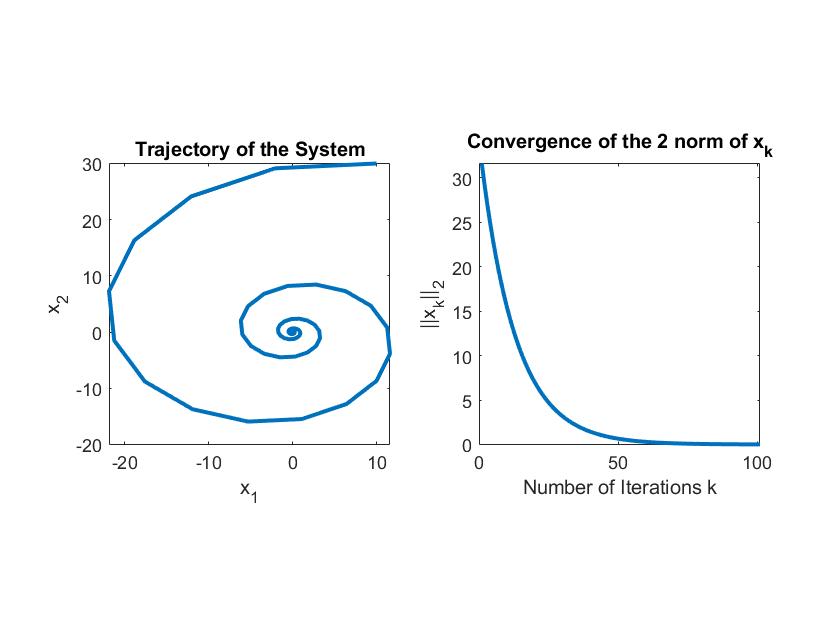}
    \caption{$l_2$: Trajectory and convergence results under constant step size $\alpha = 0.5$, $\theta=\pi/4$.}
    \label{fig:l2}
\end{figure}

\subsection{$l_\infty$ norm}
The phenomenon is the same: When $\alpha_k = \alpha$, that is constant step size, it converges to the fixed point $[0,0]^\top$.
When $\alpha_k = \frac{1}{\log k}$ and $\alpha_k = \frac{1}{\sqrt{k}}$, it also converges to fixed point.
However, when $\alpha_k=\frac{1}{k}$ (i.e., diminishing step size), there is a constant error and it can not converge to the fixed point for all $\theta \notin \Theta$.
Fig. \ref{fig:linfinity_theta_90} and Fig. \ref{fig:linfinity_theta_45} show the finite-sample bound for $\theta = \pi/2$ and $\theta=\pi/4$ respectively.
As we can see, the derived finite-sample bound \eqref{eqn:infpiover2} and \eqref{eqn:infpiover4} can approximate the real trajectory really well.
Fig. \ref{fig:linfinity_theta_120} is the trajectory and convergence results when $\theta=2\pi/3$, which also verified our obtained theoretic results \eqref{eqn:infpiover4}.
\begin{figure}[!htp]
    \centering
\includegraphics[width=\linewidth]{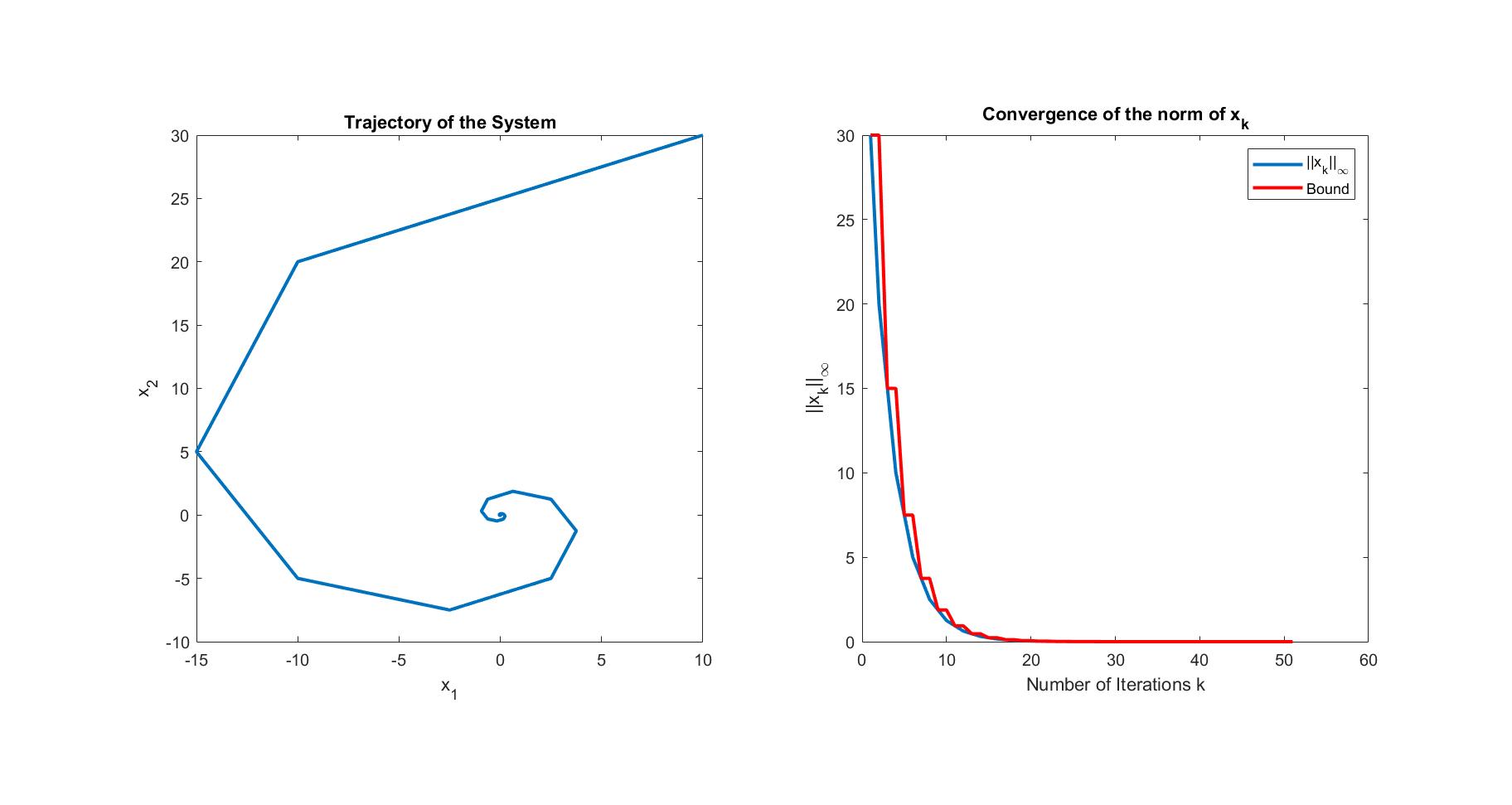}
    \caption{$l_\infty$: Trajectory and convergence results for $\theta=\frac{\pi}{2}$, $\alpha=0.5$.}
    \label{fig:linfinity_theta_90}
\end{figure}

\begin{figure}[!htp]
    \centering
 \includegraphics[width=\linewidth]{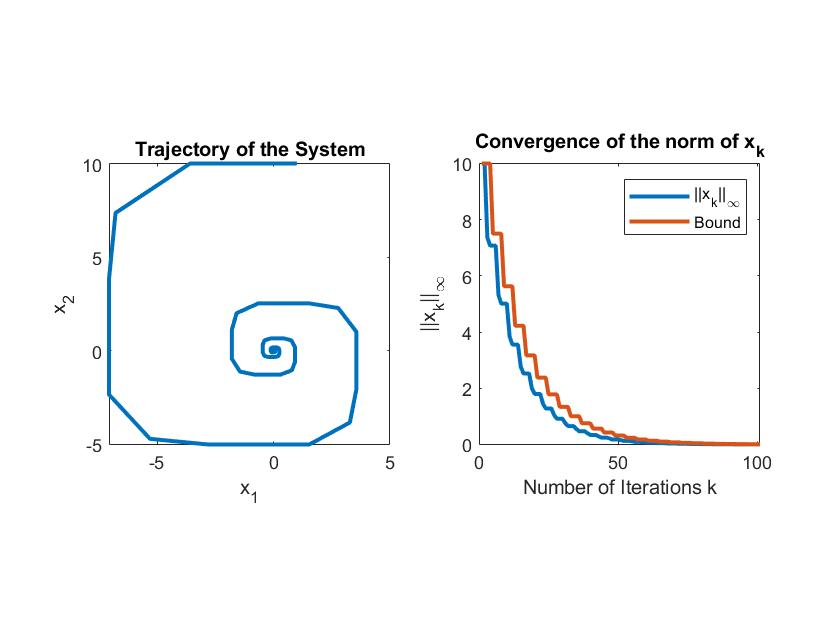}
    \caption{$l_\infty$: Trajectory and convergence results for $\theta=\frac{\pi}{4}$, $\alpha=0.5$.}
    \label{fig:linfinity_theta_45}
\end{figure}

\begin{figure}[!htp]
    \centering
    \includegraphics[width=\linewidth]{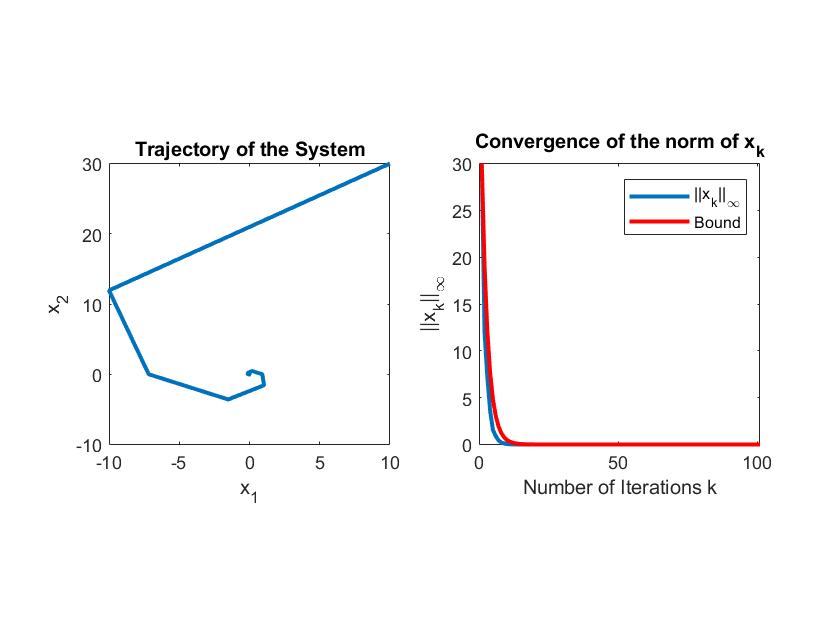}
    \caption{$l_\infty$: Trajectory and convergence results for $\theta=\frac{2\pi}{3}$, $\alpha=0.5$.}
    \label{fig:linfinity_theta_120}
\end{figure}

\subsection{$l_2$ norm with noise}
The initial position is set as $x_1 = [1,3]^\top$.
We used Gaussian noise with zero mean and variance $A= 2$.
The expectation $E[||x_k||^2]$ is calculated based on $10^4$ experiments.
Each experiment has 100 iterations.
Fig \ref{fig:l2_stoch_theta_45} shows the results of the convergence speed and the derived bound.
\begin{figure}[!htp]
    \centering
\includegraphics[width=\linewidth]{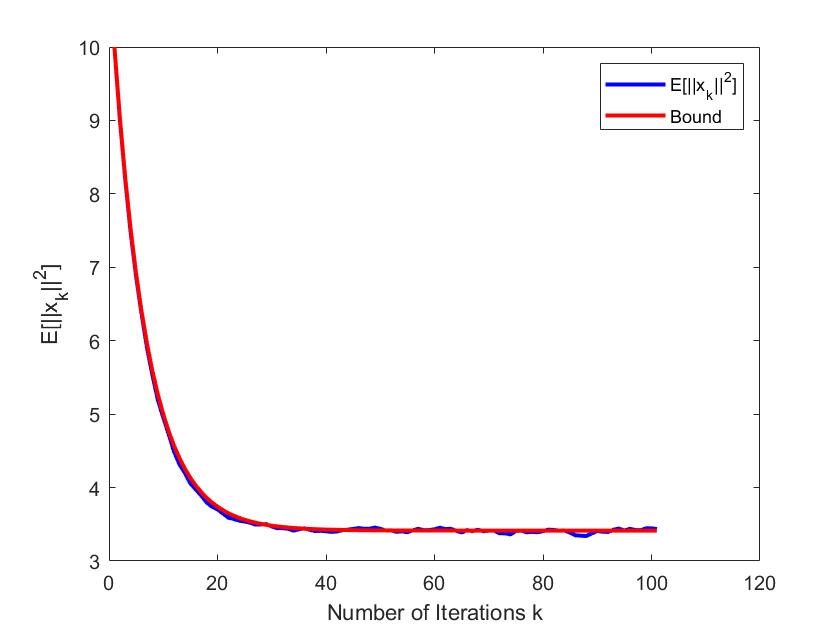}
    \caption{$l_2$ with noise: Trajectory and convergence results for $\theta=\frac{\pi}{4}$, $\alpha=0.5$, $A=2$, $B=0$.}
    \label{fig:l2_stoch_theta_45}
\end{figure}

When $\theta=\frac{\pi}{4}$ and $\alpha=0.5$, to make system stable, $B$ should satisfy $B<(1-\mu)/\alpha^2=0.5858$.
The expectation $E[||x_k||^2]$ is calculated based on $10^5$ experiments.
Each experiment has 1000 iterations.
Figure \ref{fig:l2_stoch_theta_45_B} shows the result of the convergence results when $B=0.5$.
It verifies the correctness of the derived bound \eqref{eqn:2normnoise}.
\begin{figure}[!htp]
    \centering
\includegraphics[width=\linewidth]{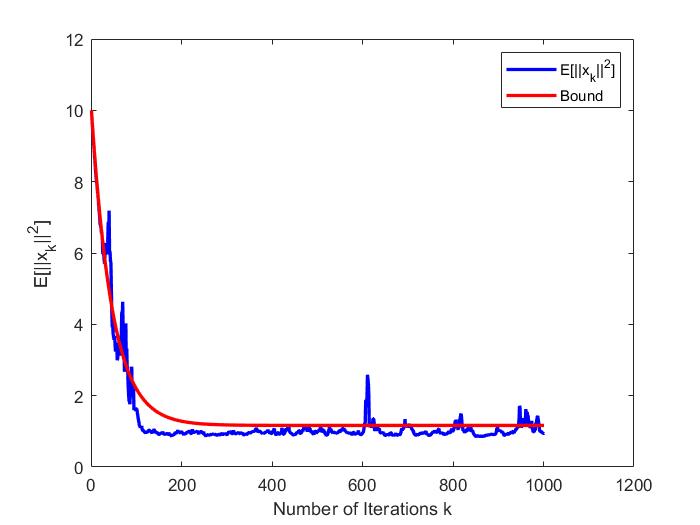}
    \caption{$l_2$ with noise: Trajectory and convergence results for $\theta=\frac{\pi}{4}$, $\alpha=0.5$, $A=0.1$, $B=0.5$.}
   \label{fig:l2_stoch_theta_45_B}
\end{figure}

\section{CONCLUSIONS}\label{sec:conclusion}
In this article, we analyzed the finite-sample bounds of rotation operators which is neither contractive nor non-expansive, under $l_2$ norm and $l_\infty$ norm with and without noise respectively.
Simulation results are provided to illustrate our theoretical results.
Even though we only considered a two-dimensional rotation matrix here, we hope this work can give some insight into the extension to any other similar operators, either linear or nonlinear.
One extension could be to find the theoretical bound for the stochastic $l_\infty$ iteration.
One other possible extension can be the convergence speed when the step size is not constant.

\bibliographystyle{./IEEEtran} 
\bibliography{./IEEEabrv,./IEEEexample}

\end{document}